\documentclass[a4paper,11pt]{article}
\usepackage{jheppub} 
\usepackage{amsthm,amssymb,tikz,tikz-cd}
\input{macros.sty}
\newcommand{\din}{\Rrightarrow}
\newcommand{\dinat}{\mathbf{Dinat}}
\newcommand{\ctr}{\mathbf{ctr}}
\newcommand{\otv}{\otimes_\cV}
\newcommand{\xc}{\mathbb{X}\mathcal{C}}
\newcommand{\Fr}{z}
\newcommand{\Fb}{\gamma}
\newcommand{\coc}[1]{\Fun(#1^\op,\cV)}
\newcommand{\hc}[3]{\underset{#1}{#3\,\mathbb H\, #2}}

\newcommand{\itk}[1]{%
   \begin{tikzpicture}[baseline=(current bounding box.center),>={Triangle[open,width=0pt 20]},scale=0.5]
   #1
   \end{tikzpicture}
   }

\arxivnumber{2412.07198} 

\makeatletter
\gdef\@fpheader{}
\makeatother

\title{Tube Category, Tensor Renormalization and Topological Holography}

\author{Tian Lan}
\affiliation{Department of Physics, The Chinese University of Hong Kong,\\ Shatin, New Territories, Hong Kong, China}
\emailAdd{tlan@cuhk.edu.hk}

\abstract{
Ocneanu's tube algebra provides a finite algorithm to compute the Drinfeld center of a fusion category. In this work we reveal the universal property underlying the tube algebra. Take a base category $\mathcal V$ which is strongly concrete, bicomplete, and closed symmetric monoidal. For physical applications one takes $\mathcal V=\mathbf{Vect}$ the category of vector spaces. Given a $\mathcal V$-enriched rigid monoidal category $\mathcal C$ (not necessarily finite or semisimple) we define the tube category $\mathbb X \mathcal C$ using coends valued in $\mathcal V$. Our main theorem established the relation between (the category of representations of) the tube category $\mathbb X \mathcal C$ and the Drinfeld center $Z(\mathcal C)$: $Z(\mathcal C)\hookrightarrow \mathrm{Fun}(\mathbb X \mathcal C^{\mathrm{op}},\mathcal V)\cong Z(\mathcal C\hookrightarrow\mathrm{Fun}(\mathcal C^{\mathrm{op}},\mathcal V))\hookrightarrow Z(\mathrm{Fun}(\mathcal C^{\mathrm{op}},\mathcal V))$. Physically, besides viewing the tube category as a version of TFT with domain being the tube, we emphasize the ``Wick-rotated'' perspective, that the morphisms in $\mathbb X \mathcal C$ are the local tensors of fixed-point matrix product operators which preserves the symmetry $\mathcal C$ in one spatial dimension. We provide a first-principle flavored construction, from microscopic quantum degrees of freedom and operators preserving the symmetry, to the macroscopic universal properties of the symmetry which form the Drinfeld center. Our work is thus a proof to the 1+1D topological holography in a very general setting.
}

\begin{document}
\maketitle
\flushbottom

\section{Introduction}

The notion of center plays a central role in the study of algebraic theories, and also finds intriguing applications in the physical study of topological holography~\cite{KW1405.5858,KWZ1502.01690,KZ1705.01087,CJK+1903.12334,KZ1905.04924,KZ1912.01760,KWZ2108.08835,XZ2205.09656,LY2208.01572,KWZ1702.00673,TW1912.02817,HC2310.16878,KZ2011.02859,KZ2107.03858,KZ2201.05726,JW1905.13279,JW1912.13492,KLW+2003.08898,KLW+2005.14178,JW2106.02069,ABE+2112.02092,BS2304.02660,BS2305.17159,CW2203.03596,CW2205.06244,LZ2305.12917,LYW2312.15958,MYLG2412.20546}. In this paper we are interested in the Drinfeld center $Z(\cC)$ of a monoidal category $\cC$, which computes the 2+1D dimensional bulk of, either a 1+1D boundary theory, or a 1+1D theory with symmetry. The Drinfeld center $Z(\cC)$ has various names in the physics literature in the context of topological holography: the background category in the enriched category approach~\cite{KZ1705.01087,CJK+1903.12334,KZ1905.04924,KZ1912.01760,XZ2205.09656,LY2208.01572,KWZ2108.08835}, topological skeleton~\cite{KZ2011.02859,KZ2107.03858,KZ2201.05726}, categorical symmetry~\cite{JW1905.13279,JW1912.13492,KLW+2003.08898,KLW+2005.14178,JW2106.02069}, symmetry TFT (topological field theory)~\cite{ABE+2112.02092,BS2304.02660,BS2305.17159}, symmetry topological order~\cite{CW2203.03596,CW2205.06244} and quantum currents~\cite{LZ2305.12917}.

It is a practically  important question how the Drinfeld center can be efficiently computed. To this end, Ocneanu~\cite{Ocn94} proposed a tube algebra approach. The original idea of Ocneanu was in the context of operator algebra. It was later generalized to arbitrary spherical fusion category $\cC$: Ref.~\cite{PSV1511.07329} proves that the category of representations of the tube algebra is equivalent to $Z(\cC)$. There are also many related studies based on the Levin-Wen string-net model~\cite{LW0404617,Kir1106.6033,KK1104.5047,LW1311.1784} or tensor networks~\cite{SWB+1409.2150,BMW+1511.08090,WBV1711.07982}. 

Refs.~\cite{HK1806.01800,Har1911.07271} introduced the notion of tube category $\cT\cC$, the horizontal categorification of the tube algebra. The category of representations of the tube algebra, is then replaced by $\Fun(\cT\cC^\op,\ve)$ which is still equivalent to $Z(\cC)$. This point of view makes the TFT nature of the tube algebra more transparent, by allowing the circle (boundary conditions of the tube) to be decorated by different objects of $\cC$.

However, the construction of tube algebra or tube category was still somehow \textit{ad hoc}. In this work, we try to reveal the universal property underlying the construction of tube category. To distinguish from the traditional definition $\cT\cC$ in the literature, as well as to emphasize the physical picture of local tensor, we denote our tube category by $\xc$. Instead of the \textit{ad hoc} graphical calculus manipulations on the tube, we instead define the morphisms of the tube category as the coend
\[ \xc(a,b):=\int^{x}\cC(ax^{RR},xb).\] The categorical structure then naturally follows from the universal property of coend.

Our treatment makes the generalization of tube category quite straightforward; we can consider the tube category of a monoidal category enriched over a more general base $\cV$, as long as the necessary universal properties still apply (for example, all the coend exists). We also do not need to assume that the category $\cC$ behaves well (for example, finite or semisimple). Our main theorem is the following:
\begin{theorem}
    Let $\cV$ be a strongly concrete, bicomplete, closed symmetric monoidal category and $\cC$ be a rigid monoidal category enriched over $\cV$. Construct the tube category $\xc$ as in Definition~\ref{def.xc}. As $\cV$-enriched categories,
    \[ Z(\cC)\hookrightarrow \Fun(\xc^\op,\cV)\cong Z(\cC\hookrightarrow\coc{\cC})\hookrightarrow Z(\Fun(\cC^\op,\cV)),\]
    where the monoidal structure of $\Fun(\cC^\op,\cV)$ is the Day convolution product~\cite{Day70}, and $Z(\cC\hookrightarrow\coc{\cC})$ denotes the relative center of the Yoneda embedding $\cC\hookrightarrow\coc{\cC}$. The above composition is moreover a braided fully faithful functor, and can be considered as the lift of the Yoneda embedding $\cC\hookrightarrow\coc{\cC}$ to the Drinfeld center.
\end{theorem}

The physical picture also becomes clear after a universal perspective of the tube category is achieved. Consider a 1+1D lattice quantum system with some symmetry. Let $\cC$ be the category of conserved charges and intertwiners of the symmetry. Then, any symmetric operator on the 1+1D lattice can be represented by a matrix product operator (MPO) whose bonds are labeled by objects in $\cC$, and local tensors are morphisms in $\cC$. One can then consider the renormalization of such symmetric MPOs. If we require the renormalization to preserve only the symmetry, but not any specific dynamics, we argue that such renormalization is mathematically formulated by a cowedge under the functor $\cC(a-,-b)$ which picks out the local tensor with virtual bonds $a,b$; it generally loses some information about the dynamics. However, we believe that the coend, the universal or ``largest'' cowedge, preserves all the universal information about the quantum 1+1D system with symmetry. Indeed, we can prove mathematically that the representations of such constructed tube category form the Drinfeld center. This way, our results explained how the Drinfeld center arise as macroscopic (long wavelength or thermodynamic limit) properties of the symmetry, starting from the microscopic quantum degrees of freedom and operators preserving the symmetry. Therefore, we give a proof, in a very general setting, to the 1+1D topological holography between 1+1D quantum system with symmetries and 2+1D topological orders.
There are works~\cite{Jon2304.00068,JNPW2307.12552,Oga2308.08087} in the operator algebra framework, also of the similar flavor of ``first principles'' which constructs the Drinfeld center from microscopic data and proves the (topological) holography in different settings.

It is worth noting that our framework does not require the symmetry to be finite, invertible, or even anomaly-free. However, for the generalized cases, e.g., the continuous $U(1)$ symmetry, $\cC=\Rep\,U(1)$, the category of representations of the tube category $\coc{\xc}\cong Z(\cC\hookrightarrow\coc{\cC})$ is much larger than the case when $\cC$ is only a fusion category. The tube category approach fails to be an efficient, or even finite, algorithm to compute the center of large $\cC$. It is also not clear whether $Z(\coc{\cC})$ can still represent a topological system for large $\cC$.

This work is organized as follows: In Section~\ref{sec.pre} we introduce necessary mathematical preliminaries and set up our notations. In Section~\ref{sec.xc} we motivate the construction of the tube category $\xc$ from the tensor renormalization of 1+1D symmetric MPOs. In Section~\ref{sec.embed} we prove the embedding $Z(\cC)\hookrightarrow \coc{\xc}$. In Section~\ref{sec.equ} we prove the equivalence $\coc{\xc}\cong Z(\cC\hookrightarrow\coc{\cC})$.

\section{Preliminaries}\label{sec.pre}
For a category $\cC$, its hom set is denoted by $\cC(a,b)$. Composition of morphisms is denoted by either $fg$, or $f\circ g$ for emphasis. $\cC^\op$ denotes the opposite category of $\cC$, i.e., $\cC^\op(a,b)=\cC(b,a)$. $\Cat$ denotes the 2-category of categories, functors and natural transformations. $\Set$ denotes the category of sets and (ordinary) maps. We use $-$ or $?$ for an unspecified variable in a map or functor.
For a monoidal category, the tensor product is denoted by $\ot$ and the tensor unit is denoted by $\one$.
For tensor product between objects,
we often omit $\ot$ and denote by simply juxtaposition $x\ot y=xy$. The right dual of an object $x$ in a rigid monoidal category is denoted by $x^R$ while the left dual is denoted by $x^L$. Our convention for left/right dual follows that of left/right adjoint functors, i.e., when doing evaluations, the left/right dual appears on the left/right, $\ev_{x^L,x}: x^L\ot x\to \one$ or $\ev_{x,x^R}: x\ot x^R\to \one$, while for coevaluations, the left/right dual appears on the right/left, $\coev_{x^R,x}: \one \to x^R\ot x $ or $\coev_{x,x^L}: \one\to x\ot x^L$. The object indices of structure morphisms ($\id,\ev,\coev,\dots$) may be dropped if they are clear from the context.
$\cV$ is a fixed category and we consider $\cV$-enriched categories or simply $\cV$-categories (see e.g.~\cite{KYZZ2104.03121}, the textbook~\cite{Kel05}). For technical simplicity, we have some assumptions on $\cV$:
\begin{definition}[Cosmos] A category $\cV$ is called a (Bénabou) \emph{cosmos} if it is
\begin{enumerate}
    \item Symmetric monoidal. The tensor product of $\cV$ is denoted by $\otv$.
    \item Closed,  i.e., $\cV$ is self-enriched: there is an object $[a,b]\in \cV$ representing the functor $\cV(-\otv a,b)$, namely $\cV(-\otv a,b)\cong \cV(-,[a,b])$.
    \item Bicomplete (complete and cocomplete), so that all limits and colimits, in particular all ends and coends, exist in $\cV$.
\end{enumerate}
\end{definition}
\begin{remark}
    Strictly speaking, one should better specify the size of diagrams with respect to which $\cV$ is bicomplete. In this paper we mainly consider coends indexed by a given category $\cC$ (which describes the physical symmetry), and thus the size of $\cV$ can be chosen according to the size of $\cC$. For example, if one is interested in fusion category symmetry, where $\cC$ is a fusion category with finitely many (isomophism classes of) simple objects and finite dimensional hom spaces, it suffices to take $\cV=\vf$ the category of finite dimensional vector spaces, which admits all finite limits and colimits. If one considers a continuous symmetry group whose (isomorphism classes of) irreducible representations are infinitely many, $\cV$ needs to be enlarged to $\ve$, the category of (possibly infinite dimensional) vector spaces, in order to ensure the existence of coends.
\end{remark}
\begin{definition}[Strongly concrete] A cosmos $\cV$ is called \emph{strongly concrete} if the hom functor ${\cV(\one,-):\cV\to\Set}$ is faithful.
\end{definition}
\begin{remark}
    Unpacking the definition, we require that the morphism map 
     \begin{align*}
        \cV(a,b)&\to \Set(\cV(\one,a),\cV(\one,b))\\ f&\mapsto f\circ-
    \end{align*}
    is injective. $\cV(a,b)$ is a subset of $\Set(\cV(\one,a),\cV(\one,b))$. We will then consider a strongly concrete cosmos $\cV$ as a concrete category with $\cV(\one,x)$ being the underlying set of $x\in \cV$. A set $X$ in the essential image of $\cV(\one,-)$, i.e., a triple $$(X\in \Set,x\in \cV,X\cong\cV(\one,x))$$ is then a structured set underlying the object $x$ in $\cV$. By abuse of notation, we will simply say that (the set) $X$ is an object in $\cV$. Under the perspective that $\cV$ is a category of structured sets and maps preserving structures, notationally we do not distinguish an object $x$ in $\cV$  from the underlying structured set $\cV(\one,x)$, as if $x=\cV(\one,x)$.
    Together with the closed condition, we know $\cV(a,b)\cong \cV(\one,[a,b])= [a,b]$. In other words, we have natural isomorphisms: $\cV(x\otv a,b)\cong \cV(x,\cV(a,b)).$
\end{remark}

In this paper, the base category $\cV$ is always assumed to be a strongly concrete cosmos. It makes no harm for physics oriented reader to take $\cV=\ve$ the category of vector spaces (over complex numbers $\C$) and linear maps; the $\ve$-enriched categories are simply linear categories and all the prefixes $\cV$- can be replaced for linear. Other possible choices include $\cV=\Set$ for ordinary categories and $\cV=\mathbf{Ab}$ for additive categories. However, the strongly concrete condition excludes choices such as $\mathbf{sVec}$ or $\Rep G$. 

Over a strongly concrete cosmos $\cV$, enriched categories can be dealt with as if they are ordinary categories. To explain this idea more precisely, we need to introduce the notion of bi(multi)-$\cV$-map as the generalization of bi(multi)-linear map. Note that the hom functor $\cV(\one,-):\cV\to \Set$ is a lax symmetric monoidal functor with the lax monoidal structure being
\[\cV(\one,x)\xt \cV(\one,y) \xrightarrow{\otv} \cV(\one\otv\one,x\otv y)\cong \cV(\one,x\otv y);\]
Viewing objects in $\cV$ as structured sets and taking the Cartesian product of them, there is then a canonical ordinary map $\otv: x\xt y\to x\otv y$:
\[ x\xt y = \cV(\one,x)\xt \cV(\one,y) \xrightarrow{\otv}\cV(\one\otv\one,x\otv y)\cong \cV(\one,x\otv y)= x\otv y.\]
\begin{definition}[Bi-$\cV$-map]
For objects $x,y,a\in \cV$, an ordinary map $f:x\xt y \to a$ is called a \emph{bi-$\cV$-map} if it factors through $x\otv y$, i.e., there exists $\tilde f:x\otv y\to a$ in $\cV$ such that $f=\tilde f \circ \otv$ as ordinary maps.  One can similarly define \emph{multi-$\cV$-maps}. 
\end{definition}
\begin{remark}
    By the closed condition and the strongly concrete condition
    \begin{align*}
    \cV(x\otv y,a) &\cong \cV(x,[y,a])
    \\
    &\hookrightarrow \Set(\cV(\one,x),\cV(\one,[y,a]))
    \\&
    \cong \Set(\cV(\one,x),\cV(y,a))
    \\&
    \hookrightarrow\Set(\cV(\one,x),\Set(\cV(\one,y),\cV(\one,a)))
    \\&
    =\Set(x,\Set(y,a))\cong \Set(x\xt y,a),
    \end{align*}
    we know that a bi-$\cV$-map $f:x\xt y \to a$ is a map in $\cV$ in each variable, in the sense:
\begin{enumerate}
    \item For each $\chi\in x$, $ f(\chi,-):y\to a$ is a map in $\cV$ and for each $\eta\in y$, $f(-,\eta):x\to a$ is also a map in $\cV$;
    \item The maps $\chi\mapsto  f(\chi,-)$ and $\eta\mapsto f(-,\eta)$ are also maps in $\cV$.
\end{enumerate} 
The canonical map $\otv: x\xt y\to x\otv y$ is the image of $\id_{x\otv y}$ under the injective map ${\cV(x\otv y,a)\hookrightarrow \Set(x\xt y,a)}$ with $a=x\otv y$. The factorization $\tilde f$ of $f$ is unique by naturality in $a$ and injectivity. Note that the braiding $x\otv y\to y\otv x$ in $\cV$ is the factorization of the bi-$\cV$-map $x\xt y\to y\xt x\to y\otv x$, and thus must be unique and essentially exchanging arguments in the Cartesian product.
\end{remark}
\begin{definition} Over a strongly concrete cosmos $\cV$,
\begin{enumerate}
    \item A \emph{$\cV$-category} is an ordinary category whose hom sets are objects in $\cV$ and compositions are multi-$\cV$-maps.
    \item A \emph{$\cV$-functor} $F:\cC\to \cD$ between $\cV$-categories $\cC,\cD$ is an ordinary functor $F:\cC\to \cD$ such that the morphism maps $F^{a,b}:\cC(a,b)\to \cD(F(a),F(b))$ are all in $\cV$. \emph{$\cV$-natural transformations} between $\cV$-functors are just ordinary natural transformations. For $\cC, \cD$ two $\cV$-categories, $\Fun(\cC,\cD)$ denotes the full subcategory of $\Cat(\cC,\cD)$ consisting of $\cV$-functors, which is automatically a $\cV$-category.
    \item Given two $\cV$-categories $\cC$ and $\cD$, their tensor product is denoted by $\cC\otv\cD$ which is still a $\cV$-category, whose objects are $\Ob(\cC\otv\cD):=\Ob(\cC)\xt\Ob(\cD)$ and morphisms are $\cC\otv\cD((x,a),(y,b)):=\cC(x,y)\otv\cD(a,b)$. The composition is given by ${(g\otv s)(f\otv r)}:={(gf)\otv(sr)}$ and the identity is given by $\id_{(x,a)}:=\id_x\otv\id_a$.
    When $\cV=\Set$, $\cC\ot_{\Set} \cD$ is just the Cartesian product $\cC\xt\cD$. There is clearly a canonical ordinary functor $\cC\xt\cD\to \cC\otv\cD$ which is identity on objects and the canonical map $\otv: \cC(x,y)\xt\cD(a,b)\to\cC(x,y)\otv\cD(a,b)$ on morphisms.
    \item Given $\cV$-categories $\cC,\cD$ and $\cE$, an ordinary functor $F:\cC\xt \cD\to \cE$ is called a \emph{bi-$\cV$-functor} if it factors through $\cC\otv\cD$, equivalently, the morphism maps are bi-$\cV$-maps. One can similar define \emph{multi-$\cV$-functors}.
    \item A \emph{(braided) monoidal $\cV$-category} $\cC$ is a category $\cC$ which is both a (braided) monoidal category and a $\cV$-category such that the tensor product $\ot:\cC\xt\cC\to \cC$ is a bi-$\cV$-functor.
\end{enumerate} 

\end{definition}

\begin{remark}
Similarly, a bi-$\cV$-functor $ F:\cC\xt\cD\to \cE$ is a $\cV$-functor in each variable:
\begin{enumerate}
    \item For each $x\in \cC$, $F(x,-): \cD\to \cE$ is a $\cV$-functor;
    \item The functor $x\mapsto F(x,-), f\mapsto F(f,-)$ is also a $\cV$-functor from $\cC$ to $\Fun(\cD,\cE)$.
\end{enumerate}
The factorization of $F$ is again unique. By a slight abuse of notation we use the same symbol for the corresponding $\cV$-functor $F:\cC\otv\cD\to\cE$. One may check that a natural transformation between $\cV$-functors $F,G:\cC\otv\cD\to \cE$ is the same as that between bi-$\cV$-functors $F,G:\cC\xt\cD\to\cE$.
\end{remark}

\begin{definition}
    [Relative center and Drinfeld center]
    Let $F:\cC\to\cD$ be a monoidal functor. The \emph{relative center} $Z(F)$ is the monoidal category defined as follows (structure morphisms are omitted for simplicity):
\begin{enumerate}
   \item An object of $Z(F)$ is a pair $(z,\gamma_{z,-})$ where $z$ is an object in $\cD$ and $\gamma_{z,x}:zF(x)\to  F(x)z$ is a collection of isomorphisms in $\cD$, natural in $x\in \cC$, satisfying the equation:
    \begin{align*}
    &\left(zF(x)F(y)\cong zF(xy)\xrightarrow{\gamma_{z,xy}}F(xy)z\cong F(x)F(y)z\right)\\
    &=\left( zF(x)F(y)\xrightarrow{\gamma_{z,x}\ot F(\id_y)}F(x)zF(y)\xrightarrow{F(\id_x)\ot \gamma_{z,y}} F(x)F(y)z\right).
    \end{align*}
    $\gamma_{z,x}$ is called the \emph{half-braiding}. The above equation is similar to half of the hexagon equations for braidings, hence the name. We still refer to the above as the \emph{hexagon equation} although its commutative diagram is not a hexagon.
    \item A morphism $f$ from $(z,\gamma_{z,-})$ to $(z',\tau_{z',-})$ is a morphism $f:z\to z'$ in $\cD$ which commutes with half-braidings: $\tau_{z',x}(f\ot F(\id_x))=(F(\id_x)\ot f)\gamma_{z,x}$ for any $x\in\cC$.
   \item The tensor product of $(z,\gamma_{z,-})$ and $(z',\tau_{z',-})$ is given by $(zz',\mu_{z z',-})$ where \begin{equation*} \mu_{z z',x}=\left(zz'F(x)\xrightarrow{\id_z\ot \tau_{z',x}}zF(x)z'\xrightarrow{\gamma_{z,x}\ot \id_{z'}}F(x)zz'\right). 
   \end{equation*}
\end{enumerate}
    The \emph{Drinfeld center} $Z(\cC)$ is the relative center of the identity functor $Z(\id_\cC)$, which is moreover braided with braiding $c_{(z,\gamma_{z,-}),(z',\tau_{z',-})}=\gamma_{z,z'}$. 
\end{definition}
\begin{remark}
When $F:\cC\to\cD$ is a $\cV$-functor, $Z(F)((z,\gamma_{z,-}),(z',\tau_{z',-}))$ can be defined as an equalizer in $\cV$. Thus, the above definition generalizes straightforwardly for monoidal $\cV$-functors, and $Z(F)$ is automatically a monoidal $\cV$-category. 
\end{remark}
\medskip

We recall the most essential notions about (co)end; for a more comprehensive exposition, see e.g. \cite{Lor1501.02503}. 
\begin{definition}
[Dinatural transformation] Let $\cC,\cD$ be two categories and $F,G:\cC^\op \xt \cC \to \cD$ two functors. A \emph{dinatural transformation} $\alpha:F\din G$ consists of a family of morphisms $\{\alpha_x: F(x,x)\to G(x,x)\}_{x\in \cC}$ such that for any morphism $f:x\to y$ in $\cC$, the following diagram commutes.
\[
\begin{tikzcd}
    & {F(x,x)} \arrow[r, "\alpha_x"] & {G(x,x)} \arrow[rd, "{G(\id_x,f)}"]  &          \\
    {F(y,x)} \arrow[rd, "{F(\id_y,f)}"'] \arrow[ru, "{F(f,\id_x)}"] &                                &                                      & {G(x,y)} \\
    & {F(y,y)} \arrow[r, "\alpha_y"] & {G(y,y)} \arrow[ru, "{G(f,\id_y)}"'] &    
\end{tikzcd}
\]
\end{definition}
The set of dinatural transformations from $F$ to $G$ is denoted by $\dinat(F,G)$. 
One can check that the compositions of a dinatural transformation $\alpha:F\din G$ with natural transformations $\beta: G\Rightarrow G'$ and $\mu: F'\Rightarrow F$, defined by $(\beta\circ\alpha)_x:=\beta_{x,x}\alpha_x,$ and $(\alpha\circ\mu)_x:=\alpha_x\mu_{x,x}$ are still dinatural transformations. Denote by $\Delta_x:\cC^\op\xt \cC\to \cD$ the constant functor taking value at the object $x\in \cD$. It is clear that $\cD(x,y)\cong \Cat(\cC^\op\xt\cC, \cD)(\Delta_x,\Delta_y)$ ($\Delta$ is a fully faithful functor from $\cD$ to $\Cat(\cC^\op\xt\cC, \cD)$). 

\begin{definition}
    [(Co)end]
    Fix a functor $F:\cC^\op\xt \cC\to\cD$. A \emph{wedge} over $F$ is an object $d\in \cD$ together with a dinatural transformation $\alpha: \Delta_d\din F$. The \emph{end} of $F$, denoted by $(\End F,\pi:\Delta_{\End F}\din F)$ or simply $\End F$, is the universal wedge, i.e., it is the wedge such that the composition with $\pi$ is a bijection:
    \begin{align*}
        \cD(d,\End F) &\cong \dinat(\Delta_d, F)\\
        f &\mapsto \pi\circ \Delta_f.
    \end{align*}
     Denote the inverse map by $\alpha\mapsto \bar\alpha$. More explicitly, for any wedge $\alpha: \Delta_d\din F$ there is a unique morphism $\bar \alpha:d\to \End F$ such that $\alpha=\pi\circ\Delta_{\bar\alpha}$.

    Dually, a cowedge under $F$ is a dinatural transformation $F\din \Delta_d$ and the coend of $F$, $\Coend  F$, is the universal cowedge. End and coend are also denoted by the integral notation
    \[ \End F=\int_{x\in \cC} F(x,x), \quad \Coend F=\int^{x\in \cC} F(x,x).\]
    Due to the universal property, end and coend are functorial; given natural transformation $\mu:F'\Rightarrow F$ we also denote the corresponding morphism by the integral notation
    \[ \int^x \mu_{x,x}: \int^x F'(x,x)\to \int^x F(x,x).\]
    The index category $\cC$ in the integral will be omitted when no confusion arises. 
\end{definition}
\begin{remark} (Co)ends are special cases of (co)limits while (co)limits are also special cases of (co)ends. The $\cV$-functor category $\Fun(\cC,\cD)$ is automatically enriched in $\cV$, due to the fact that the set of natural transformations $\Fun(\cC,\cD)(F,G)$ is given by the end of $\cD(F(-),G(-)):\cC^\op\xt\cC\to \cV$,
    \[\Fun(\cC,\cD)(F,G) \cong \int_x \cD(F(x),G(x)),\]    
    which is automatically an object in $\cV$.
\end{remark}
\begin{remark}
For $\cV$-functors $F,G:\cC^\op\otv\cC\to\cD$, we define dinatural transformations between them as those between the corresponding bi-$\cV$-functors $F,G:\cC^\op\xt\cC\to\cD$. The (co)wedges and (co)ends of $\cV$-functors can then be defined similarly. 
\end{remark}
From now on, all the categories and functors are by default enriched over a fixed strongly concrete cosmos $\cV$ and we will omit the prefixes $\cV$- in most cases.

\section{Fixed-point local tensors and tube category \texorpdfstring{$\xc$}{XC}}\label{sec.xc}
We motivate the definition of tube category in this section. Firstly, we need to generalize several notions in tensor network theory:

\begin{definition}
\label{def.tensor}
    Let $\cC$ be a monoidal $\cV$-category. 
    \begin{enumerate}
    \item 
    We view a sequence of objects $x_1,x_2,\dots,x_n\in\cC$ as a one-dimensional \emph{lattice}; an object $x_i$ in the sequence is then called a \emph{(local lattice) site}. An \emph{operator} in $\cC$ on the lattice $x_1,\dots,x_n$ is an endo-morphism in $\cC(\bigotimes_{i=1}^n x_i,\bigotimes_{i=1}^n x_i)$.
    \item
   A \emph{local tensor} in $\cC$ on a local site $x\in \cC$ with \emph{virtual bonds} $a,b\in \cC$ is a morphism in $\cC(ax,xb)$. The local site $x$ is also referred to as the \emph{physical bond}. Let the physical bond vary over objects in $\cC$ while keeping the virtual bonds $a,b$ fixed, we get a $\cV$-functor $\cC(a-,-b):\cC^\op\otv\cC\to \cV$. The (horizontal) \emph{tensor contraction} is defined to be the following natural transformation $\ctr^{a,b,c}: \cC(a-,-b)\otv \cC(b?,?c)\Rightarrow \cC(a-?,-?c)$:
   \[
    \begin{tikzcd}
    \ctr^{a,b,c}_{x,x',y,y'}:{\cC(ax,x'b)\otv \cC(by,y'c)} \arrow[rrr, "(-\ot \id_y)\otv (\id_{x'} \ot -)"] & & & {\cC(axy,x'by)\otv \cC(x'by,x'y'c)} \arrow[d, "\circ"] \\
                                                               & & & {\cC(axy,x'y'c)}                                   
    \end{tikzcd}
   \]
   Graphically,
   \[\ctr^{a,b,c}_{x,x',y,y'}(\itk{
    \node (l) at (-2,-1) {};
    \node[rectangle,draw] (m) at (0,0) {$f$};
    \node(r) at (2,1) {};
    \draw 
      (0,-2)node[below]{$x$} -- (m) -- (0,2)node[above]{$x'$}
      (l)--node[above]{$a$} (m) -- node[above]{$b$} (r);
   }\otv 
   \itk{
    \node (l) at (-2,-1) {};
    \node[rectangle,draw] (m) at (0,0) {$g$};
    \node(r) at (2,1) {};
    \draw 
      (0,-2)node[below]{$y$} -- (m) -- (0,2)node[above]{$y'$}
      (l)--node[above]{$b$} (m) -- node[above]{$c$} (r);
   }
   )=\itk{
    \node (l) at (-2,-1) {};
    \node[rectangle,draw] (m) at (0,0) {$f$};
    \node[rectangle,draw] (r) at (2,1) {$g$};
    \node (rr) at (4,2) {};
    \draw 
      (0,-2)node[below]{$x$} -- (m) -- (0,3)node[above]{$x'$}
      (2,-2)node[below]{$y$} -- (r) -- (2,3)node[above]{$y'$}
      (l)--node[above]{$a$} (m) -- node[above]{$b$} (r) -- node[above]{$c$} (rr);
   }\ .
   \]
   \item
   A \emph{matrix product operator} (MPO) $O$ in $\cC$ on the lattice $x_1,\dots,x_n$ with open virtual bonds $a,b$ is a morphism $O\in \cC(a\left(\bigotimes_{i=1}^n x_i\right),\left(\bigotimes_{i=1}^n x_i\right)b)$ that can be represented in the form of repeated tensor contractions of local tensors, i.e.,
   \[ O=\itk{
    \node[rectangle,draw] (m1) at (0,0) {$f_{1}$};
    \node[rectangle,draw] (m2) at (2,0.2) {$f_{2}$};
    \node[rectangle,draw] (mm) at (8,0.8) {$f_{n-1}$};
    \node[rectangle,draw] (mn) at (10,1) {$f_{n}$};
    \node (l) at (-2,-0.2) {};
    \node (r) at (12,1.2) {};
    \draw 
      (0,-2)node[below]{$x_{1}$} -- (m1) -- (0,3)node[above]{$x_{1}$}
      (2,-2)node[below]{$x_{2}$} -- (m2) -- (2,3)node[above]{$x_{2}$}
      (8,-2)node[below]{$x_{n-1}$} -- (mm) -- (8,3)node[above]{${x_{n-1}}$}
      (10,-2)node[below]{$x_{n}$} -- (mn) -- (10,3)node[above]{$x_{n}$}
      (l)--node[above]{$a$} (m1) --  (m2) -- (4,0.4)
      (6,0.6) -- (mm) -- (mn) -- node[above]{$b$} (r);
    \node at (5,0.5){$\ldots$};
   }\ .\]
   Indices of contracted virtual bonds have been omitted. When the dangling open virtual bonds are trivial, $a=b=\one$, $O$ is an operator on $x_1,\dots,x_n$.
   \end{enumerate}
\end{definition}
\begin{remark}
    When $\cC=\vf$, we are back to ordinary tensor networks. An ordinary lattice is just a sequence of local vector spaces; ordinary operators, local tensors and MPOs are just ordinary linear maps between tensor product of multiple (physical or virtual) vector spaces. In this case, a local tensor in $\vf$ can be viewed as a ``matrix of local operators". More precisely, take $f\in \vf(ax,xb)$, we choose a basis $\{\alpha_i\}\subset a= \vf(\C,a)$ and a dual basis $\{\beta_j\}\in \vf(b,\C)$ of the virtual bonds, then 
    \[ f_{ij}:=\itk{
    \node[rectangle,draw] (l) at (-2,-1) {$\alpha_i$};
    \node[rectangle,draw] (m) at (0,0) {$f$};
    \node[rectangle,draw] (r) at (2,1) {$\beta_j$};
    \draw 
      (0,-2)node[below]{$x$} -- (m) -- (0,2)node[above]{$x$}
      (l)--node[above]{$a$} (m) -- node[above]{$b$} (r);
   }=(\id_x\ot \beta_j)f(\alpha_i\ot \id_x)\in \vf(x,x)\]
    is a local operator on the site $x$; a local tensor is equivalent to a matrix whose rows and columns are indexed by bases of virtual bonds and matrix elements are local operators. Moreover, matrix elements of the tensor contraction is first taking tensor product of matrix elements and then summing over the basis of the contracted virtual bond: $\forall f\in \vf(ax,xb), g\in \vf(by,yc),$
    \[ \left(\ctr^{a,b,c}_{x,x,y,y}(f\ot g)\right)_{ik}=\sum_j f_{ij}\ot g_{jk},\]
    which is like performing a matrix multiplication. Hence the name ``matrix product operator."
\end{remark}
\begin{remark}
    If there is a monoidal functor $F:\cC\to \ve$, lattices, local tensors and MPOs in $\cC$ can be mapped to ordinary ones using $F$. The pair $(\cC,F)$ plays the role of selecting a special class of vector spaces and linear maps; the kinematics of the resulting lattice system (in the image of $F$) is constrained. It is then natural to think that $(\cC,F)$ imposes a \emph{symmetry} on 1+1D lattice systems. A classical example is taking $\cC=\Rep G$ the category of representations of some group $G$, and $F$ the forgetful functor. Then the vector spaces allowed are those carrying $G$-symmetry charges ($G$-representations); the linear maps allowed are the $G$-symmetric operators (intertwiners). For a generic $\cC$, we will also physically interpret objects of $\cC$ as symmetry charges, morphisms in $\cC$ as symmetric operators (or tensors), and $\cC$ itself as the \emph{charge category}. In the pair $(\cC,F)$, $\cC$ plays the role of a charge category, and $F:\cC\to\ve$ provides a concrete realization of $\cC$ on ordinary lattice systems. Admitting such a fiber functor (a monoidal functor to $\ve$ or sometimes more strictly to $\vf$) is the physical condition for an abstract symmetry $\cC$ to be \emph{anomaly-free}~\cite{TW1912.02817,KLW+2005.14178,LYW2312.15958,MYLG2412.20546}. However, in this paper our treatments are all carried out within $\cC$; there is no need to require the existence of a fiber functor. Our work thus applies to anomalous symmetries as well. 
\end{remark}

From now on by local tensors or MPOs we always mean those in $\cC$, and occasionally we may use \emph{symmetric} local tensor and \emph{symmetric} MPO for emphasis. Physically it suffices to study only the MPOs for 1+1D systems, guaranteed by the following:
\begin{lemma}
   Suppose $\cC$ has (left) duals. Any operator in $\cC$ can be rewritten as an MPO. 
\end{lemma}
\begin{proof}
    Without losing generality, let's consider an operator $O\in \cC(xy,xy)$ on two adjacent sites $x,y$. The strategy to construct an MPO is simply by bending the physical legs. For example, let 
    $$f=\itk{
    \node[rectangle,draw, minimum width=5ex] (o) at (-0.5,0) {$O$};
    \node (xl) at (-1,2) {$x$};
    \node (x) at (-1,-2) {$x$};
    \node (a) at (0,2) {};
    \draw (x) -- (o.south-|x)
    (2,1)node[right]{$y$} to[out=-170,in=90] (o.north-|a)
    (o.north-|xl) -- (xl)
    (o.south-|a) to[out=-90,in=-170] (2,0)node[right]{$y^L$}
    ; }
    =O(\id_x\ot \coev_{y,y^L})\in \cC(x,xyy^L)$$ and $$g=
    \itk{
     \node (yl) at (4,-2) {$y$};
     \node (y) at (4,2.5) {$y$};
    \draw (2,1)node[left]{$y$} to[out=10,in=-90] (4,2)-- (y)
    (2,0)node[left]{$y^L$} to[out=10,in=90] (4,0) -- (yl);
    }=
    \id_y\ot \ev_{y^L,y}\in \cC(yy^Ly,y),$$
    one can easily check that
    \[ O=\itk{
    \node[rectangle,draw, minimum width=5ex] (o) at (-0.5,0) {$O$};
    \node (xl) at (-1,2.5) {$x$};
    \node (x) at (-1,-2) {$x$};
    \node (a) at (0,2) {};
    \draw (x) -- (o.south-|x)
    (2,1) to[out=-170,in=90] (o.north-|a)
    (o.north-|xl) -- (xl)
    (o.south-|a) to[out=-90,in=-170] (2,0);
    \node (yl) at (4,-2) {$y$};
     \node (y) at (4,2.5) {$y$};
    \draw (2,1) to[out=10,in=-90] (4,2)-- (y)
    (2,0) to[out=10,in=90] (4,0) -- (yl);}
    =
    \ctr^{\one,yy^L,\one}_{x,x,y,y}(f\otv g).\] Note that the choice of left or right duals here is just a matter of convention depending on how we orient the virtual bonds. 
\end{proof}
\begin{remark}
    The way to rewrite an operator as an MPO is in general not unique. The naive decomposition in the above proof takes the virtual bond $yy^L$. Similar choices for a long lattice $x_1,\dots,x_n$ would lead to very large virtual bonds, which is physically not desirable. If the morphisms in $\cC$ admit good factorizations, the virtual bonds can be chosen in a more canonical way. In ordinary tensor networks, the singular value decomposition is often utilized. More generally, assume that $\cC$ admits image decomposition, i.e., for any morphism $f:a\to b$, there exists an object $\im f$, an epimorphism $e:a\to \im f$ and a monomorphism $m:\im f\to b$, such that $f=me$. Then an operator $O\in \cC(xy,xy)$ can be decomposed along the horizontal direction, and becomes an MPO in a canonical way. More precisely, one can 
    \begin{enumerate}
        \item bend the legs to get the mate 
    \begin{align*}
    \tilde O&=
    \itk{
    \node[rectangle,draw, minimum width=10ex] (o) at (0,0) {$O$};
    \node (xl) at (-3,-2) {$x^L$};
    \node (x) at (-1,-2) {$x$};
    \node (yl) at (3,2) {$y^L$};
    \node (y) at (1,2) {$y$};
    \draw (x) -- (o.south-|x)
    (y) -- (o.north-|y)
    (o.north-|x) arc(0:180:1) -- (xl)
    (o.south-|y) arc(180:360:1) -- (yl)
    ; }
    \\
    &=(\ev_{x^L,x}\ot\id_{yy^L})(\id_{x^L}\ot O\ot \id_{y^L})(\id_{x^Lx}\ot \coev_{y,y^L})\in \cC(x^Lx, yy^L), 
    \end{align*}
    \item then take the image decomposition of $\tilde O$: \[\tilde O=x^Lx\xrightarrow[]{e} \im \tilde O\xrightarrow[]{m}yy^L,\]
    graphically,
    \[ \itk{
    \node[rectangle,draw, minimum width=10ex] (o) at (0,0) {$O$};
    \node (xl) at (-3,-2) {$x^L$};
    \node (x) at (-1,-2) {$x$};
    \node (yl) at (3,2) {$y^L$};
    \node (y) at (1,2) {$y$};
    \draw (x) -- (o.south-|x)
    (y) -- (o.north-|y)
    (o.north-|x) arc(0:180:1) -- (xl)
    (o.south-|y) arc(180:360:1) -- (yl)
    ; }=
    \itk{
    \node[rectangle,draw] (e) at (0,-0.5) {$e$};
    \node[rectangle,draw] (m) at (0,1.5) {$m$};
    \node (xl) at (-1,-2) {$x^L$};
    \node (x) at (1,-2) {$x$};
    \node (yl) at (1,3) {$y^L$};
    \node (y) at (-1,3) {$y$};
    \draw (x) -- (e) -- (xl)
    (y) -- (m) -- (yl) (e)--node[right]{$\im\tilde O$}(m)
    ; },
    \]
    \item finally bend the legs back to obtain an MPO 
    \begin{align*}
    O&=
    \itk{
    \node[rectangle,draw, minimum width=10ex] (o) at (0,0) {$O$};
    \node (xl) at (-5,2) {$x$};
    \node (x) at (-1,-2) {$x$};
    \node (yl) at (5,-2) {$y$};
    \node (y) at (1,2) {$y$};
    \draw (x) -- (o.south-|x)
    (y) -- (o.north-|y)
    (o.north-|x) arc(0:180:1) -- (-3,0) arc(360:180:1) -- (xl)
    (o.south-|y) arc(180:360:1) -- (3,0) arc(180:0:1)-- (yl)
    ; }
    \\
    &=\itk{
    \node[rectangle,draw] (e) at (-1,-0.5) {$e$};
    \node[rectangle,draw] (m) at (1,1.5) {$m$};
    \node (xl) at (-3,3) {$x$};
    \node (x) at (-1,-2) {$x$};
    \node (yl) at (3,-2) {$y$};
    \node (y) at (1,3) {$y$};
    \draw (x) -- (e) to[out=-135,in=-90] (xl)
    (y) -- (m) to[out=45,in=90] (yl) (e)--node[above left]{$\im\tilde O$}(m)
    ; }
    \\
    &=\ctr^{\one,\im\tilde O,\one}_{x,x,y,y}\Big(\big(e(\coev_{x,x^L}\ot \id_x)\big)\otv \big((\id_y\ot\ev_{y^L,y})m\big)\Big)
    \end{align*}
    (see also~\cite{LZ2305.12917}).
    \end{enumerate} Such canonically determined virtual bond $\im \tilde O$ (which is of course an object in the charge category $\cC$) can be understood as the charge transported by $O$. It further leads to the idea of quantum current~\cite{LZ2305.12917} that a symmetric MPO describes how the symmetry charge (virtual bonds) flows in the one-dimensional space, with the local tensors interpreted as how the symmetry charge interacts quantum mechanically with the local sites as it flows.
\end{remark}
\bigskip

In this work we are interested in the universal properties of symmetric operators in 1+1D lattice systems with symmetries. Since any symmetric operator can be written as a symmetric MPO, and any MPO essentially consists of local tensors, it should suffice to study the local tensors. As we care about universal properties, we should not restrict to particular lattices; instead, we should study local tensors on arbitrary local sites. Therefore, our strategy is to take $\cC(ax,xb)$ for arbitrary $x\in \cC$, and then try to ``integrate out" the local site $x$. In the next several paragraphs, we will argue that
\begin{center}
\emph{
The coend 
$\int^x \cC(ax,xb)$
is the universal space of fixed-point local tensors.}
\end{center}
\medskip


We first review renormalizations in the usual sense. Usual renormalization schemes generally depend on the dynamics of the physical systems; for a quantum lattice model, the dynamics is controlled by the Hamiltonian, which is a Hermitian operator. Therefore, for the moment, we restrict to the case $\cC=\Hilb$, the category of Hilbert spaces, which is a $\dag$-category, with the $\dag$-structure being the usual Hermitian conjugate. This is only for the purpose of motivation; later we will drop the requirement for the $\dag$-structure.

Given a lattice $x_1, x_2, \dots, x_n$ in $\Hilb$ and meanwhile a Hamiltonian $H\in \Hilb(\ot_i x_i,\ot_i x_i)$, we consider a typical renormalization scheme. The first step is to ``coarse-grain'' the lattice sites by combining several adjacent sites into one. For notational simplicity and without losing generality, we can just take $x_i$ as the tensor product of several adjacent lattice sites in the original lattice model, and then represent the coarse-graining by linear maps $p_i: x_i\to y_i$, where $y_i$ is the coarse-grained lattice site, and $y_1,\dots,y_n$ is the coarse-grained lattice. We assume that $p_i$ is surjective and $p_i^\dag:y_i\to x_i$ is an isometry which preserves the inner product, i.e., $p_ip_i^\dag=\id_{y_i}$. Denote by $P=\ot_i p_i:\ot_i x_i\to \ot_i y_i$ the linear map between total Hilbert spaces. We refer to $P$ or $p_i$ as the \emph{lattice renormalization} map. For a chosen $P$, there is naturally an \emph{operator renormalization} map
\begin{align*}
\pi_P: \Hilb(\ot_i x_i,\ot_i x_i)&\to \Hilb(\ot_i y_i,\ot_i y_i)\\
O &\mapsto \pi_P(O)=POP^\dag.
\end{align*}

After the coarse-graining, one needs to determine the effective Hamiltonian on the coarse-grained lattice. A natural choice is to take
\[ H_\eff=\pi_P(H)=PHP^\dag.\]
At this step, the dynamics must kick in; the lattice renormalization $P=\ot_i p_i$ has to be carefully chosen so that the effective Hamiltonian $H_\eff$ can correctly represent the low energy behavior of the original Hamiltonian $H$. Such requirement translates to that
\begin{itemize}
    \item $P^\dag H_\eff P=P^\dag P H P^\dag P$ must act within the low energy eigen-subspace of $H$;
    \item or equivalently, $P^\dag P$ projects to the low energy eigen-subspace of $H$.
\end{itemize}
We will loosely refer to the conditions above as \emph{preserving the (low-energy) dynamics}.

The above renormalization scheme does not fit the purpose of this paper --- studying the universal behavior of lattice models with symmetries. We would not like to restrict to a specific lattice, nor a specific Hamiltonian; instead, we care about all the possible lattices and all the possible symmetric Hamiltonians. For such purpose, we would like to generalize the idea of renormalization:
\begin{enumerate}
    \item  Any map $p_i:x_i\to y_i$ should be qualified as a lattice renormalization, in the sense that there exists a Hamiltonian (for example, just the sum of local projections $H=-\sum_i p_i^\dag p_i$) such that $P=\ot_i p_i$ preserves the low-energy dynamics. 
    \item We also want to reformulate the above scheme to get rid of the $\dag$-structure. Consider the map $h:=HP^\dag$. As $P^\dag$ is assumed to be an isometry, $h$ is essentially the renormalized Hamiltonian. It is not a usual operator on a fixed lattice, but rather an ``operator'' from the coarse-grained lattice to the original lattice; the latter is indeed a better formulation in the context of renormalization as we need to constantly deform the physical systems. To get the effective Hamiltonian $H_\eff$, as a usual operator on the coarse-grained lattice, from $h$, one choice is to post-compose the lattice renormalization
    \[ Ph;\]
    the other choice is to first pre-compose the lattice renormalization and then apply the operator renormalizatoin map
    \[ \pi_P(hP).\]
    It is clear that the two choices agree and both define the effective Hamiltonian
    \begin{equation}
        Ph=PHP^\dag=H_\eff=PHP^\dag PP^\dag=\pi_P(hP). \label{eq.eff}\tag{\#}
    \end{equation}
    This property \eqref{eq.eff} can be interpreted in the following way: after operator renormalization to the coarse-grained lattice, $h$ commutes with the lattice renormalization $P$; such condition is essential in our later generalizations.
    \item Finally, we would like to allow the lattice renormalization and operator renormalization to be independent. In the above usual scheme, the lattice renormalization $P$ should be chosen against a specific Hamiltonian $H$, and the operator renormalization $\pi_P$ is then chosen against $P$. Such logic fails when no specific dynamics is given. We thus want to study independent operator renormalizations against arbitrary lattice renormalizations; neither of them are required to preserve any specific dynamics.
\end{enumerate}

Since in 1+1D any operator can be expressed as an MPO, we will study the local \emph{tensor renormalization} instead of the operator renormalization. We will see that the cowedge $\alpha: \cC(a-,-b)\din \Delta_X$ (taking $\cV=\ve$) exactly fits our generalized idea of renormalization. For each $x\in \cC$, there is a linear map $\alpha_x: \cC(ax,xb)\to X$. We first explain that these linear maps $\alpha_x$ can be understood as tensor renormalizations.
Explicitly, the condition for these linear maps to form a dinatural transformation is that
\[\begin{tikzcd}
{\cC(ay,xb)} \arrow[rr, "{\cC(af,xb)}"] \arrow[d, "{\cC(ay,fb)}"'] && {\cC(ax,xb)} \arrow[d, "\alpha_x"] \\
{\cC(ay,yb)} \arrow[rr, "\alpha_y"]                                        & &X                                 
\end{tikzcd}\]
which means that for any $h:ay\to xb$ and any symmetric operator $f:x\to y$, 
\[ \alpha_x\Big(h(\id_a\ot  f)\Big)=\alpha_y\Big((f\ot \id_b)h\Big). \]
Graphically,
\begin{equation}
\alpha_x\Big(
   \itk{
    \node (l) at (-2,-1) {};
    \node[rectangle,draw] (m) at (0,0) {$h$};
    \node[rectangle,draw] (mu) at (0,-2) {$f$};
    \node(r) at (2,1) {};
    \draw 
      (0,-3)node[below]{$x$} -- (mu) --node[right]{$y$} (m)--(0,1)node[above]{$x$}
      (l)--node[above]{$a$} (m) -- node[above]{$b$} (r);
   }
   \Big)=\alpha_y\Big(
  \itk{
    \node (l) at (-2,-1) {};
    \node[rectangle,draw] (m) at (0,0) {$h$};
    \node[rectangle,draw] (mu) at (0,2) {$f$};
    \node(r) at (2,1) {};
    \draw 
      (0,-1)node[below]{$y$} -- (m) --node[left]{$x$}(mu)-- (0,3)node[above]{$y$}
      (l)--node[above]{$a$} (m) -- node[above]{$b$} (r);
   }
   \Big)
\label{eq.cowedge}\tag{$\ddag$}    
\end{equation}
The reader may spot that the dinatural condition \eqref{eq.cowedge} is a generalization of the condition \eqref{eq.eff}, under the following interpretation:
We view $f:x\to y$ as a lattice renormalization and $h$ as an MPO (with open virtual bonds $a,b$) acting between potentially different lattices. The vector space $X$ is viewed as an abstract space of local tensors, and $\alpha_x,\alpha_y$ as tensor renormalization maps. Then, $h(\id_a\ot f)$ is applying the MPO $h$ after a lattice renormalization $f$, while $(f\ot\id_b)h$ is applying the MPO $h$ before a lattice renormalization $f$. The results are local tensors on potentially different lattice sites $x,y$, which can not be compared directly.  After applying tensor renormalizations $\alpha_x,\alpha_y$, the results can be compared in $X$. Therefore, we interpret the dinatural condition as requiring that, after tensor renormalization (applying $\alpha$), the local tensor is invariant before and after an arbitrary lattice renormalization $f$. In short, the image of $\alpha$ in $X$ are renormalized local tensors, which are at fixed-point under lattice renormalizations. 

Note that here we imposed only the constraints regarding the symmetry (requiring all tensors to be morphisms in $\cC$). No dynamical information has been taken into account. Therefore, the renormalization linear map $\alpha_x$ in a cowedge is in general ``wilder'' than a usual tensor renormalization that is required to preserve dynamics as we discussed before. To conclude, a cowedge describes a limiting space of tensors under wild renormalizations which are required to preserve only symmetry but not any dynamics; dynamical information may be lost.
To fix such issue, we take the coend of $\cC(a-,-b)$, i.e., the universal cowedge. For any cowedge $X$ under $\cC(a-,-b)$, there exists a unique map $\int^x\cC(ax,xb)\to X$ compatible with the dinatural transformations (renormalizations to either $\int^x\cC(ax,xb)$ or $X$); thus the coend $\int^x\cC(ax,xb)$ is like a ``supremum": it loses the least information among all the wild renormalizations (cowedges). We may draw a line to arrange various renormalization schemes, depending on how much information is maintained at fixed-point:
\medskip
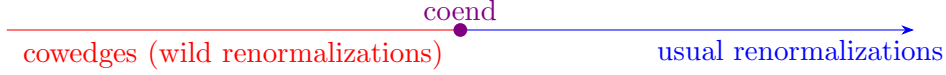
\begin{center}
    \itk{
    \draw[red]  (0,0)--node[below]{cowedges (wild renormalizations)}  (12,0)node[violet,above]{coend};
    \draw[blue,-Stealth] (12,0) -- (18,0) --node[below]{usual renormalizations} (24,0);
    \node[circle, fill, scale=0.5,violet] at (12,0) {};
    }
\end{center}
\begin{itemize}
    \item The usual renormalization takes both the symmetry and the dynamics into account, so they are on the right;
    \item The cowedges (wild renormalizations) takes only the symmetry but not the dynamics into account, so they are on the left;
    \item The coend is the rightmost cowedge.
\end{itemize}
It is an interesting question whether there are always usual renormalization schemes which approaches the coend from the right; the answer is unclear to us but we wishfully expect so. In this work we will be satisfied working with the coend.

\begin{remark}\label{remark.trace} 
    It is inspiring to consider the example of ordinary tensor networks (i.e., no symmetry, taking $\cC=\vf$ as a linear monoidal category). The dinatural condition \eqref{eq.cowedge} is also similar to the typical property of the ordinary trace of linear maps: For $x,y\in \vf$ two vector spaces (not necessarily $x=y$), any two linear maps $U:x\to y$ and $V:y\to x$ satisfy $\Tr_x(VU)=\Tr_y(UV)$. Indeed, one can easily verify that the coend $\int^{x} \vf(x,x)$ is given by $\C$ with the universal dinatural transformation being the ordinary trace function $\Tr_x:\vf(x,x)\to \C$. More generally, $\int^{x}\vf(ax,xb)$ is given by $\vf(a,b)$ with the universal dinatural transformation being the partial trace. 
    Partial trace agrees with our above discussions in its renormalization features. Let's consider, for simplicity, the partial trace on a bipartite system. Suppose that (the vector space of) the total system is $A\ot B$, which consists of (the vector spaces of) two subsystems $A$ and $B$. Taking partial trace over subsystem $A$ is a linear map $\Tr_A:\vf(A\ot B,A\ot B)\to \vf(B,B)$. The partial trace $\Tr_A$ ``integrates out'' the information in $A$, with a largest possible space $\vf(B,B)$ as its target space, such that the measurements (expectation values) of operators supported on $B$ are all preserved. Therefore, one can view the coend $\int^x\cC(ax,xb)$ as a generalized partial trace, which takes the symmetry into account. This perspective also partially justified the notation $\int^x$, in that the coend is indeed a generalization, or categorification, of summation, integration, or taking trace. Thus, the coend $\int^x\cC(ax,xb)$ indeed ``integrates out'' the local site $x$.
\end{remark}
\bigskip

\begin{remark}
    We used the graphical convention for 1+1D operators (tensors) that the spatial direction is horizontal and the time direction is vertical. The virtual bonds are along the spatial direction and the physical bonds are along the time direction. Another reasonable interpretation to the cowedge $\alpha: \cC(a-,-b)\din \Delta_X$ is to consider a ``Wick-rotation,'' i.e., for $h:ay\to xb$, think $x,y$ as virtual bonds (along the spatial direction) while $a,b$ as physical bonds (along the time direction). Graphically, one can take \eqref{eq.cowedge} and rotate by $90^\circ$ (switching spatial and time directions). In such picture, the dinatural transformation $\alpha_x$ or $\alpha_y$ should be understood as gluing the virtual bond $x$ or $y$, respectively, and then the physical bonds $a,b$ are the degrees of freedom on a spatial topology of a ring $S^1$, while $\alpha_x\Big(h(\id_a\ot  f)\Big)=\alpha_y\Big((f\ot\id_b)h\Big)$ is an operator acting on $S^1$ which might involve a lattice renormalization that changes the physical bond from $a$ to $b$. We can see that the dinatural condition of a cowedge properly enforces the periodic boundary condition on a ring, in the sense that any $f:x\to y$ can be moved across the periodic boundary to the other side. The spacetime topology of such Wick-rotated local tensors is clearly a tube $S^1\xt I$. We will see that the coend of $\cC(a-,-b)$ allows us to define a new category $\xc$ in a very natural way, which recovers the tube category in the literature when $\cC$ is a spherical fusion category.
\end{remark}

\begin{remark}
Dually, for a wedge $\alpha: \Delta_X\din \cC(a-,-b)$, $\alpha_x:X\to\cC(ax,xb)$ picks out tensors which exactly commute with all local symmetric perturbations. The renormalization picture (projecting out higher energy operators) seems more general, and it is less general to require local tensors to exactly commute with perturbations (for example, if we are considering a continuous spectrum). For the Wick-rotated picture, the dinatural condition for wedge also fails to satisfy the periodic boundary condition. Although the universal property of end does not naturally lead to the tube category, it is $\End \cC(a-,-b)$ that is mostly considered in the literature (see e.g.~\cite{HXZ2404.01162,DS07}). Technically, this may be due to the fact that the hom functor preserves limits, thus using the end (assuming rigidity)
\[ \int_x \cC(ax,xb)\cong \cC(a,\int_x xbx^L)\cong \cC(\int^x x^L a x,b),\]
one directly obtain representable functors, and further more functors $\int^x x^L ? x$ and $\int_x x?x^L$ out of $\cC$, which turns out to map into $Z(\cC)$ and are left/right adjoint to the forgetful functor $Z(\cC)\to \cC$. Under nice conditions, e.g. $\cC$ is finite semisimple, the end and coend do coincide; however, we emphasize that it is $\Coend \cC(a-,-b)$ that allows a much more natural physical interpretation as well as mathematical construction.
\end{remark}

\bigskip

Based on pervious discussions, we now give the key definition in this work    
\begin{definition}
    [Tube category $\xc$] \label{def.xc} Let $\cC$ be a rigid monoidal $\cV$-category. We define the $\cV$-category $\xc$, called the \emph{tube category} of $\cC$, as follows:
    \begin{enumerate}
    \item The category $\xc$ has the same objects as $\cC$, $\Ob(\xc)=\Ob(\cC)$.
    \item The hom spaces are defined to be
    \[ \xc(a,b):=\Coend \cC(a-^{RR},-b)=\int^x \cC(ax^{RR},xb).\] The corresponding universal dinatural transformation is denoted as $\pi^{a,b}_x:\cC(ax^{RR},xb)\to \int^y\cC(ay^{RR},yb).$
   \item The composition is induced by tensor contraction and renormalization (taking coend), 
   \[  
    \begin{tikzcd}
    \int^x \cC(ax^{RR},xb) \otv \int^y \cC(by^{RR},yc)
    \arrow[rrr,"\int^{x,y} \ctr^{a,b,c}_{x^{RR},x,y^{RR},y}"]
    &&&\int^{x,y}\cC(ax^{RR}y^{RR},xyc)
    \arrow[d]\\
    &&& \int^x\cC(ax^{RR},xc)
    \end{tikzcd}
    \]
    where the last arrow is uniquely determined by viewing the coend $(\int^x\cC(ax^{RR},xc),\pi^{a,c})$ as a cowedge under $$\cC(a-^{RR}?^{RR},-?c):(\cC\otv\cC)^\op\otv(\cC\otv\cC)\to \cV$$ with dinatural transformation $\beta$, whose components are $$\beta_{(x,y)}=\pi^{a,c}_{xy}:\cC(ax^{RR}y^{RR},xyc)\to \int^z\cC(az^{RR},zc).$$  
    More explicitly, take $f:ax^{RR}\to xb$ and $g:by^{RR}\to yc$, the composition in $\xc$ is 
    \[ \pi^{b,c}_y(g)\circ \pi^{a,b}_x(f)=\pi^{a,c}_{xy}(\ctr^{a,b,c}_{x^{RR},x,y^{RR},y}(f\otv g)).\]
    \item The identity morphism in $\xc(a,a)$ is given by $\pi^{a,a}_\one(\id_a).$
    \end{enumerate}
    
\end{definition}

\begin{remark} We explain the physical meanings of Definition~\ref{def.xc}:
\begin{enumerate}
    \item 
 The physical meaning of objects of $\xc$ is that they label the open virtual bonds of MPOs, and thus also symmetry charges flowing into and out of regions on which the MPOs are supported. 
 \item On the hom spaces we have added a double (right) dual. This is for the purpose of mathematical generality. Such defined tube category is naturally related to the Drinfeld center without further requirements. For physical applications, we may assume that $\cC$ is pivotal, so that $x^{RR}\cong x$, and the hom spaces $\xc(a,b)$ are indeed the spaces of fixed-point local tensors $\int^x\cC(ax,xb)$. Based on such perspective, we will also call the tube category $\xc$ as the category of fixed-point local tensors.
 \item Physically we can see that the composition in the tube category is exactly the horizontal tensor contraction of fixed-point local tensors.
\end{enumerate}
\end{remark}
\bigskip

We need to further pass from $\xc$ to the functor category $\coc\xc$; the main motivations are explained below:
\begin{enumerate}
    \item $\coc\xc$ is the \emph{free cocompletion} of $\xc$; in particular it contains the idempotent completion of $\xc$.
    
We have explained that the coend $\int^x \cC(ax,xb)$ is physically the space of fixed-point local tensors. Coend is a special type of colimit and we believe that in general colimit computes certain kinds of renormalization. 
Although $\xc$ above is a well-defined category, it is in general not complete or cocomplete. We have not considered the renormalization of the virtual bonds $a,b,\dots\in\Ob(\cC)=\Ob(\xc)$. In particular, if a local tensor is invariant under horizontal tensor contraction followed by renormalization, by definition (the renormalization of) such a local tensor is an idempotent in $\xc$. The retract (image) of an idempotent is an absolute (co)limit. We should at least include all the retracts, i.e., form the idempotent completion of $\xc$.  More generally, we would like to ``freely add colimits to 
$\xc$'' and study the free cocompletion, i.e., the functor category $\coc{\xc}$. Physically we consider $\coc{\xc}$ as the complete space of fixed-point local tensors. 

\item $\coc\xc$ provides a perspective of TFT with defects.

Using the ``Wick-rotated'' picture, objects in $\xc$ label physical systems on a ring $S^1$ and $\int^x \cC(ax^{RR},xb)$ are operators transforming from $a$ to $b$. We can thus consider manifolds $S^1$ decorated by objects in $\xc$ and manifolds $S^1\times I$ decorated by morphisms in $\xc$. Both decorations can be viewed as defects on the manifolds. The functor category $\coc{\xc}$ (for $\cV=\ve$) then admits a TFT-flavored interpretation, i.e., a functor $F:\xc^\op\to\ve$ assigns vector spaces $F(a)$, $F(b)$ to the ring $S^1$ decorated by $a,b$ and assigns linear maps $F(f):F(b)\to F(a)$ to the tube $S^1\times I$ decoreated by $f\in \int^x \cC(ax^{RR},xb)$.

\item $\coc\xc$ consists of the ``sectors of operators.''

Recall that a unital associative algebra $A$ can be viewed as a category $BA$ with a single object, and the category of left $A$-representations is $\Rep A=\Fun(BA,\ve)$. If we view a linear category as a ``horizontal categorification'' of an algebra by allowing more than one object, $\coc{\xc}$ can then be viewed as the category of (right) representations of $\xc$. 
$\coc{\xc}$ is closely related to the Drinfeld center $Z(\cC)$, as we will discuss in the next section. In the enriched category proposal for topological holography~\cite{KWZ2108.08835}, it is expected that the ``sectors of non-local symmetric operators'' form the Drinfeld center and serve as the background category of enrichment. Recall that a superselction sector (of states) is a representation of the local operator algebra. We believe that the appropriate definition for ``sectors of operators'' should be representations of local tensors; this is exactly $\coc\xc$, as $\xc$ is the category of fixed-point local tensors.
\end{enumerate}

\section{The embedding \texorpdfstring{$Z(\cC)\hookrightarrow \coc{\xc}$}{from Drinfeld center}}\label{sec.embed}
The following theorem can be thought of as a corollary of the more general theorem in the next section. As its proof can be put into a more intuitive graphical form, we elaborate in this section. One can see intuitively how quantum current MPOs~\cite{LZ2305.12917}, i.e., MPOs whose local tensors are the half-braidings in the Drinfeld center $Z(\cC)$, form the representations of local tensors.
\begin{theorem}
    There is a fully faithful functor $Y:Z(\cC)\to \coc{\xc}$.
\end{theorem}
\begin{proof}
We would like to show that any object $(z,\gamma)\in Z(\cC)$ determines a functor $Y_{(z,\gamma)}: \xc^\op\to \cV$. On objects, $Y_{(z,\gamma)}a:=\cC(a,z)$. On morphisms, we define $$Y_{(z,\gamma)}^{a,b}:\int^x \cC(ax^{RR},xb) \to \cV(\cC(b,z),\cC(a,z)),$$
to be the unique map in $\cV$ determined by the following dinatural transformation, constructed using the half-braiding,
\begin{align*}
   \tau^{(z,\gamma);a,b}_x :\cC(ax^{RR},xb) &\to \cV(\cC(b,z),\cC(a,z))\\
   h &\mapsto \tau^{(z,\gamma);a,b}_x(h),
\end{align*}
where $\tau^{(z,\gamma);a,b}_x(h)$ is the map in $\cV$ that takes $f\in \cC(b,z)$ to the following composition:
\[
a\xrightarrow{\id\ot\coev}ax^{RR}x^R
\xrightarrow{h\ot\id}xbx^R
\xrightarrow{\id\ot f\ot \id}xzx^R
\xrightarrow{\id\ot \gamma}xx^Rz\xrightarrow{\ev\ot\id} z,
\]
which can be expressed graphically as
\[  \tau^{(z,\gamma);a,b}_x(h)(f)=    
\itk{
    \node (l) at (-2,-1) {$a$};
    \node[rectangle,draw] (m) at (0,0) {$h$};
    \node[rectangle,draw] (r) at (2,1) {$f$};
    \node (rr) at (4,2) {};
    \draw 
      (0,-2) -- (m) -- (0,3) node[left]{$x$} arc(180:0:2) -- (rr)
      (l) -- (m) -- node[above]{$b$} (r) -- (6,3) node[right]{$z$}
      (0,-2)node[left]{$x^{RR}$} -- (0,-3) arc(180:360:2) --node[right]{$x^R$} (rr);
   }.
\]
Thus $\cV(\cC(b,z),\cC(a,z))$ becomes a cowedge under $\cC(a-^{RR},-b)$ and the desired map is uniquely determined via $Y_{(z,\gamma)}^{a,b}\pi^{a,b}_x=\tau^{(z,\gamma);a,b}_x$. Using the hexagon equation of $\gamma$, it is straightforward to show that these maps on morphisms preserve composition and thus $Y_{(z,\gamma)}$ is a functor. Moreover, any morphism $(z,\gamma)\to (z',\gamma')$ determines a natural transformation $Y_{(z,\gamma)}\Rightarrow Y_{(z',\gamma')}$ by post composition, and we have a functor $$Y: Z(\cC)\to \coc{\xc}.$$
Note that there is an obvious functor $L:\cC\to\xc$ which is identity on objects and $\pi^{a,b}_\one$ on morphisms:
\begin{align*}
L:\cC &\to \xc
\\
a&\mapsto L(a)=a
\\
L^{a,b}:\cC(a,b)&\xrightarrow{\pi^{a,b}_\one}\int^x\cC(ax^{RR},xb)=\xc(L(a),L(b)).
\end{align*}
It is easy to see that $Y_{(z,\gamma)}L= \cC(-,z)\in \coc{\cC}$ is representable. By the Yoneda lemma, any natural transformation $\beta:Y_{(z,\gamma)}\to Y_{(z',\gamma')}$ uniquely determines a morphism $\tilde\beta:z\to z'$ via $\beta_{L(a)}=\beta_a=\tilde\beta\circ- $ viewed as a natural transformation in $\coc{\cC}$. The condition for $\beta$ to be natural between $Y_{(z,\gamma)}$ and $ Y_{(z',\gamma')}$ in $\coc{\xc}$ reads,
\[ Y_{(z',\gamma')}^{a,b}\pi^{a,b}_x(h)(\tilde \beta \circ f)=\tilde \beta \circ Y_{(z,\gamma)}^{a,b}\pi^{a,b}_x(h)( f).\]
Observe that by taking $h=x\Fr x^Rx^{RR}\xrightarrow{\id\ot\ev}x\Fr $ and $f=\id_\Fr$ (this special choice is also important later)
\begin{align*}
    &Y_{(z,\gamma)}^{x\Fr x^R,\Fr }\pi^{x\Fr x^R,\Fr }_x(x\Fr x^Rx^{RR}\xrightarrow{\id\ot\ev}x\Fr )(\id_{\Fr })
    \\&=\itk{
    \node (l) at (-2,-1) {$\Fr $};
    \node (lu) at (-2,-0.3) {$x$};
    \node (ld) at (-2,-1.7) {$x^R$};
    \node (rr) at (4,2) {};
    \draw 
      (lu) to[out=26.58,in=-90] (0,2) 
      -- (0,3)  arc(180:0:2) -- (rr)
      (0,-2)  to[out=90,in=26.58]  (ld) 
      (l) -- (6,3) node[right]{$z$}
      (0,-2) -- (0,-4) arc(180:360:2) --node[right]{$x^R$} (rr);
   }=\itk{
    \node (l) at (-2,-1) {$\Fr $};
    \node (lu) at (-2,-0.3) {$x$};
    \node (rr) at (4,2) {};
    \draw 
      (lu) to[out=26.58,in=-90] (0,2) 
      -- (0,3)  arc(180:0:2) -- (rr)
      (l) -- (6,3) node[right]{$z$}
      (4,-1)--node[right]{$x^R$} (rr);}
      \\&
      =(xzx^R\xrightarrow{\id\ot\gamma}xx^Rz\xrightarrow{\ev\ot\id}z)\in \cC(x\Fr x^R,\Fr ).
\end{align*}
One immediately checks that $\tilde\beta$ commutes with half-braidings. Therefore, $Y$ is fully faithful. 
\end{proof}

\begin{remark}
     The composition $$Z(\cC)\xrightarrow{Y}\coc{\xc}\xrightarrow{-\circ L} \coc{\cC}$$ is identified with the forgetful functor composed with the Yoneda embedding $$Z(\cC)\to \cC\hookrightarrow\coc{\cC}.$$
     In other words, if we identify $Z(\cC)$ with its image under $Y$ and identify $\cC$ with its image under the Yoneda embedding, the pullback by $L$ (i.e., $-\circ L$) coincides with the forgetful functor.
\end{remark}

We would like to characterize the image of $Y$.
\begin{theorem}\label{thm.imY}
$Z(\cC)$ is equivalent to the full subcategory of $\coc{\xc}$ of objects $F$ such that $FL$ is representable.
\end{theorem}
\begin{proof}
Now suppose $F\in \coc{\xc}$ such that $FL$ is representable, $FL\cong \cC(-,\Fr )$, then one can determine a half-braiding on the object $\Fr \in\cC$ based on the data of $F$. Since $L$ is identity on objects, we know that $F(a)=FL(a)\cong \cC(a,\Fr )$. In the following derivation we will implicitly use this isomorphism, but suppress it in the notation as if $F(a)=\cC(a,\Fr )$.  Note that giving a morphism $F^{a,b}:\int^x\cC(ax^{RR},xb)\to\cV(\cC(b,\Fr ),\cC(a,\Fr ))$ is the same as giving a dinatural transformation $\cC(ax^{RR},xb)\xrightarrow{\pi^{a,b}_x}\int^y\cC(ay^{RR},yb)\xrightarrow{F^{a,b}}\cV(\cC(b,\Fr ),\cC(a,\Fr ))$. To recover the half-braiding, we take the image of
\begin{align*} &(x\Fr x^Rx^{RR}\xrightarrow{\id\ot\ev}x\Fr )=\itk{
    \node (l) at (-2,-1) {$\Fr $};
    \node (lu) at (-2,-0.3) {$x$};
    \node (ld) at (-2,-1.7) {$x^R$};
    \node (r) at (2,1) {$\Fr $};
    \draw 
      (l) -- (r) 
      (lu) to[out=26.58,in=-90] (0,2) node[above]{$x$}
       (0,-3) node[below]{$x^{RR}$} -- (0,-2) to[out=90,in=26.58]  (ld) ;
   }\in \cC(x\Fr x^Rx^{RR},x\Fr )
\\&
\xrightarrow{F^{x\Fr x^R,\Fr }\pi^{x\Fr x^R,\Fr }_x}\cV(\cC(\Fr ,\Fr ),\cC(x\Fr x^R,\Fr ));
\end{align*}
then further pick $\id_{\Fr }\in \cC(\Fr ,\Fr )$, and the mate of
\[F^{x\Fr x^R,\Fr }\pi^{x\Fr x^R,\Fr }_x(x\Fr x^Rx^{RR}\xrightarrow{\id\ot\ev}x\Fr )(\id_{\Fr })\in\cC(x\Fr x^R,\Fr ) \]
in $\cC(\Fr x^R,x^R\Fr )$ will be the desired half-braiding, which we still denote by $\Fb$. More explicitly, we define
\[ \gamma_{z,x^R}:=(zx^R\xrightarrow{\coev\ot \id} x^Rxzx^R\xrightarrow{\id \ot F\pi(x\Fr x^Rx^{RR}\xrightarrow{\id\ot\ev}x\Fr )(\id_{\Fr })}x^Rz).\]
Making use of the functoriality of $F$, 
 i.e., for any $f:ax^{RR}\to xb$ and $g:by^{RR}\to yc$, 
    \[ F^{a,b}\pi_x^{a,b}(f) F^{b,c}\pi^{b,c}_y(g)=F^{a,c}\pi^{a,c}_{xy}(\ctr^{a,b,c}_{x^{RR},x,y^{RR},y} (f\otv g)),\]
together with the fact that $FL\cong \cC(-,\Fr )$ which means for any $g:a\to b$ and $f:b\to \Fr $,
    \[ F^{a,b}\pi^{a,b}_\one(g)(f)=f\circ g,\]
one can compute the composition of $\Fb$ with other morphisms. More explicitly, for any $g: a \to x\Fr x^R$,
\begin{align*}
&F^{x\Fr x^R,\Fr }\pi^{x\Fr x^R,\Fr }_x(\itk{
    \node (l) at (-2,-1) {$\Fr $};
    \node (lu) at (-2,-0.3) {$x$};
    \node (ld) at (-2,-1.7) {$x^R$};
    \node (r) at (2,1) {$\Fr $};
    \draw 
      (l) -- (r) 
      (lu) to[out=26.58,in=-90] (0,2) node[above]{$x$}
      (0,-3) node[below]{$x^{RR}$} -- (0,-2) to[out=90,in=26.58]  (ld) ;
       }
    )(\id_{\Fr }) \circ g
   =F^{a,\Fr }\pi^{a,\Fr }_x(
   \itk{
    \node[rectangle,draw] (l) at (-2,-1) {$g$};
    \node (r) at (2,1) {$\Fr $};
    \draw 
      (l) -- (r) 
      (l.north east) to[out=26.58,in=-90] (0,2) node[above]{$x$}
      (0,-3) node[below]{$x^{RR}$} -- (0,-2) to[out=90,in=26.58]  (l.east) 
      (-4,-2)node[left]{$a$} -- (l);
   }
   )(\id_{\Fr })
   .
\end{align*}
With this trick, it becomes straightforward to show that
\begin{itemize}
\item The dinaturality of $\pi$ implies the naturality of $\Fb$. 
\item The hexagon equation of $\Fb$ is satisfied.
\item It is also easy to check that $\Fb_{\Fr ,\one}=\id_{\Fr }$. Together with naturality, hexagon equation and rigidity, it implies that $\Fb$ is invertible.
\end{itemize}
Therefore, $\Fb$ is a half-braiding. One can then check that for the $(z,\gamma)$ constructed from $F$ as above,
\[  Y_{(\Fr ,\Fb)}\cong F,\]
which amounts to verify that $Y_{(z,\gamma)}^{a,b}=F^{a,b}$, i.e., for any $h: ax^{RR}\to xb$ and $f: b\to z$,
\begin{align*}
&\itk{
    \node (l) at (-2,-1) {$a$};
    \node[rectangle,draw] (m) at (0,0) {$h$};
    \node[rectangle,draw] (r) at (2,1) {$f$};
    \node[rectangle,draw] (rr) at (2,4) {$F^{x\Fr x^R,\Fr }\pi^{x\Fr x^R,\Fr }_x(x\Fr x^Rx^{RR}\xrightarrow{\id\ot\ev}x\Fr )(\id_{\Fr })$};
    \draw 
      (0,-2) -- (m) --node[left]{$x$} (rr.south-|m)
      (l) -- (m) -- node[above]{$b$} (r) --node[right]{$z$} (rr.south-|r) 
      (rr) -- (2,6) node[above]{$z$}
      (0,-2) -- (0,-3) arc(180:360:2) --node[right]{$x^R$} (rr.south-|4,0);
   }
   =F^{a,z}\pi^{a,z}_x(\itk{
    \node (l) at (-2,-1) {$a$};
    \node[rectangle,draw] (m) at (0,0) {$h$};
    \node[rectangle,draw] (r) at (2,1) {$f$};
    \draw 
      (0,-2) -- (m) -- (0,2)node[above]{$x$}
      (l) -- (m) -- node[above]{$b$} (r) --node[above]{$z$}(4,2)
      (0,-2)node[below]{$x^{RR}$} -- (m);
   })(\id_\Fr)
   \\&
   = F^{a,b}\pi^{a,b}_x(\itk{
    \node (l) at (-2,-1) {$a$};
    \node[rectangle,draw] (m) at (0,0) {$h$};
    \draw 
      (0,-2) -- (m) -- (0,2)node[above]{$x$}
      (l) -- (m) -- node[above]{$b$} (2,1)
      (0,-2)node[below]{$x^{RR}$} -- (m);
   }
   )F^{b,z}\pi^{b,z}_\one(f)(\id_\Fr)
   =F^{a,b}\pi^{a,b}_x(h)(f).
\end{align*}
The above proves that $Y$ is essentially surjective on the full subcategory of $\coc{\xc}$ of objects $F$ such that $FL$ is representable.
\end{proof}

\begin{corollary}
   If any functor in $\coc{\cC}$ is representable, i.e., the Yoneda embedding $\cC\hookrightarrow\coc{\cC}$ is an equivalence (e.g. when $\cC$ is finite semisimple), then $${Y:Z(\cC)\to \coc{\xc}}$$ is also an equivalence.
\end{corollary}

\begin{remark}
If $\cC$ is finite semisimple, one has isomorphisms natural in $a,b$
\[ \xc(a,b)=\int^x \cC(ax^{RR},xb)\cong \bigoplus_{x\in\mathrm{Irr}(\cC)} \cC(ax^{RR},xb)\cong \int_x\cC(ax^{RR},xb),\]
where $\mathrm{Irr}(\cC)$ denotes the finite set of (representatives of isomorphism classes of) simple objects in $\cC$. This recovers the tube category in the literature. To recover the tube algebra, one can just take $P=\bigoplus_{x\in \mathrm{Irr}(\cC)}x$ and consider the algebra $\xc(P,P)$.
For a spherical fusion category $\cC$, the idempotent completion of $\xc$ is also equivalent to either $Z(\cC)$ or $\coc{\xc}$ (see e.g.~\cite{Har1911.07271}). Note that an idempotent in $\xc$ is just a local tensor which is invariant under tensor contraction and renormalization; it corresponds to an object in $Z(\cC)$, and was called a (steady) quantum current in \cite{LZ2305.12917}. 
\end{remark}
\begin{remark}
Let $\tilde L$ be the composition $\cC\xrightarrow{L} \xc \hookrightarrow \coc{\xc}$.  If $\tilde L$ lies in the image of $Y$, we think $Y^{-1}\tilde L$ as a functor $\cC\to Z(\cC)$, and by a standard derivation using the Yoneda lemma,
\[\coc\xc(\xc(-,L(a)),F)\cong FL(a)\cong \coc\cC(\cC(-,a),FL),\]one can see that $Y^{-1}\tilde L$ is left adjoint to the forgetful functor $Z(\cC)\to \cC$, which explains the choice of our notation. The condition for $\tilde L$ to fall in the image of $Y$ translates to that $\tilde L(a) L=\xc(-,a):\cC^\op\to\cV$ is representable for any $a\in \cC$. It is the case when the hom functor of $\cC$ commutes with coends (when every object in $\cC$ is projective, e.g., when $\cC$ is finite semisimple) since $\tilde L(a) L=\xc(-,a)= \int^x \cC(-x^{RR},xa)\cong \int^x\cC(-,xax^R)\cong \cC(-,\int^x xax^R).$ By Theorem~\ref{thm.imY}, we know  that $\int^x xax^R$ admits a natural half-braiding. Indeed, the functor $a\mapsto \int^x xax^R$, coinciding with $Y^{-1}\tilde L$, is also left adjoint to the forgetful functor $Z(\cC)\to \cC$.
\end{remark}



\section{The equivalence with the relative center of the Yoneda embedding}\label{sec.equ}

We first review the Yoneda Lemma and co-Yoneda Lemma in terms of end and coend. Let $\cC$ be a $\cV$-category and $G:\cC^\op\to \cV$ a $\cV$-functor. The Yoneda Lemma asserts isomorphisms 
\[ G(a)\cong \coc{\cC}(\cC(-,a),G)\cong \int_x \cV(\cC(x,a), G(x))
\]
natural in $a\in \cC$
. Similarly, the co-Yenoda Lemma asserts the natural isomorphism
\[ G(a)\cong \int^x \cC(a,x)\otv G(x).\]
One can see that in the integral notation, the hom functor $\cC(a,x)$ behaves like a ``delta functor''; for any natural transformation $\beta:G\Rightarrow G'$, we also have
\[ \beta_a=\left(G(a)\cong \int^x \cC(a,x)\otv G(x) \xrightarrow{\int^x \cC(a,\id_x)\otv \beta_x} \int^x \cC(a,x)\otv G'(x)\cong G'(a)\right).\]

Now assume that $\cC$ is monoidal. It is possible to equip $\coc{\cC}$ with a monoidal structure such that the Yoneda embedding $\cC\hookrightarrow\coc\cC$ is monoidal. Such monoidal structure is known as the Day convolution product, defined for $H,G:\cC^\op\to \cV$,
\begin{align*}
 H\star G:&=\int^{x,y} \cC(-,xy)\otv H(x)\otv G(y).
\end{align*}
Assume in addition that $\cC$ is rigid, we further have
\[ H\star G\cong \int^x H(-x^R)\otv G(x)\cong \int^x H(x)\otv G(x^L-).\]

It is then meaningful to talk about the Drinfeld center $Z(\coc\cC)$ and the relative center $Z(\cC\hookrightarrow\coc{\cC})$. An object $(F,\tau_{F,-})$ in the Drinfeld center $Z(\coc\cC)$ is a functor $F:\cC^\op\to\cV$ together with half-braiding $\tau_{F,G}:F\star G\cong G\star F$ whose component on $a\in \cC$ is
\[(\tau_{F,G})_a:\int^x F(ax^R)\otv G(x)\to \int^x F(x^La)\otv G(x).\]
An object in the relative center $Z(\cC\hookrightarrow\coc{\cC})$ is a functor $H:\cC^\op\to \cV$ equipped with a half-braiding $\gamma_{H,-}$ with respect to all representable functors, i.e.
\[(\gamma_{H,b})_a:
\int^x H(ax^R)\otv \cC(x,b)\cong H(ab^R)\to \int^x H(x^La)\otv \cC(x,b)\cong H(b^L a).\]
We list the hexagon equation in this case explicitly:
\begin{align*}
&\bigg(H(a(bc)^R)\xrightarrow{(\gamma_{H,bc})_a}H((bc)^La)\bigg)
\\&=\bigg(H(ac^Rb^R)\xrightarrow{(\gamma_{H,b})_{ac^R}}H(b^Lac^R)\xrightarrow{(\gamma_{H,c})_{b^La}}H(c^Lb^La)\bigg).
\end{align*}
Moreover, $\gamma_{H,-}$ can be extended to be with respect to all functors in $\coc{\cC}$, by defining 
\[ (\tilde\gamma_{H,G})_a:=\left(\int^x H(ax^R)\otv G(x)\xrightarrow{\int^x (\gamma_{H,x})_a\otv G(\id_x)} \int^x H(x^La)\otv G(x)\right).\]
This way, $(H,\tilde \gamma_{H,-})$ is an object in $Z(\coc\cC)$. Such assignment also defines a canonical braiding on the relative center $Z(\cC\hookrightarrow\coc\cC)$, and we have a braided embedding
\[ Z(\cC\hookrightarrow\coc\cC)\hookrightarrow Z(\coc\cC).\]

We are then ready to prove
\begin{theorem}
    $\coc\xc$ is equivalent to $Z(\cC\hookrightarrow\coc\cC)$.
\end{theorem}

\begin{proof}

Now consider $F:\xc^\op\to\cV$. $FL$ is then a functor in $\coc{\cC}$ which is identical to $F$ on objects, but on morphisms $FL$ remembers only the action of $F$ on the image of $\pi_\one$ and forgets all others. We will show that the ``extra data'' in $F$ about how it maps morphisms in $\xc(a,b)$, equips $FL$ with a half-braiding with respect to representable functors under the Day convolution product.

For any $x,a\in \cC$, we construct
\[ (\gamma_{F,x})_a=F^{x^La,ax^R}\pi^{x^La,ax^R}_{x^L}(x^Lax^R\xrightarrow{\id}x^Lax^R): FL(ax^R) \to FL(x^La). \]
Using the functoriality of $F$ together with the dinaturality of $\pi$, it is straightforward to verify that $(\gamma_{F,x})_a$ is natural in $a,x$ and the inverse to $(\gamma_{F,x})_a$ is
\[ 
 F^{ax^R,x^La}\pi^{ax^R,x^La}_{x}(ax^Rx^{RR}\xrightarrow{\id\ot\ev}a\xrightarrow{\coev\ot\id}xx^La).
\]
The hexagon equation for $\gamma_{F,-}$ follows directly from the functoriality of $F$.
Since the half-braiding $\gamma_{F,-}$ is defined using the action of $F$ on morphisms in $\xc$, natural transformations in $\coc\xc$ automatically commute with half-braidings and the assignment
\begin{align*}
    \coc{\xc}&\to Z(\cC\hookrightarrow\coc{\cC})
    \\
    F &\mapsto (FL,\gamma_{F,-})
\end{align*}
is a functor.

Conversely, suppose $(H,\gamma_{H,-})$ is an object in $Z(\cC\hookrightarrow\coc\cC)$.
We use $\gamma_{H,-}$ to lift $H$ to a functor $\hat H$ in $\coc\xc$. Define $\hat H(a)=H(a)$; for $h:ax^{RR}\to xb$, define maps dinatural in $x$ by
\begin{align*}
h&\mapsto 
\bigg( H(b)\xrightarrow{H(x^Lax^{RR}\xrightarrow{\id\ot h}x^Lxb\xrightarrow{\ev\ot \id}b)} H(x^Lax^{RR})
\\&
\xrightarrow{(\gamma_{H,x^R})_{x^La}}H(xx^La)\xrightarrow{H(a\xrightarrow{\coev\ot\id}xx^La)}H(a)\bigg),
\end{align*}
whose dinaturality follows from the naturality of $(\gamma_{H,b})_a$ in both $a,b$. 
Then \[ \hat H^{a,b}:\xc(a,b)\to \cV(\hat H(b),\hat H(a)),\]
is uniquely determined. Functoriality of $\hat H$ follows from the functoriality of $H$, the hexagon equation, and the naturality of $\gamma$;
this fact is the most challenging to check in this proof and we give more details. Take two arbitrary morphisms $f: y^L b y^{RR}\to c$ and $g: x^L a x^{RR}\to b$, the functoriality of $\hat H$ amounts to the following commutative diagram:
\[
\begin{tikzcd}
H(c) \arrow[rd, "H(f\circ(y^L g y^{RR}))"'] \arrow[r, "H(f)"] & H(y^L b y^{RR}) \arrow[rr, "{(\gamma_{H,y^R})_{y^Lb}}"] \arrow[d, "H(y^Lgy^{RR})"]                                         &  & H(yy^Lb) \arrow[rr, "H(\coev_y b)"] \arrow[d, "H(yy^Lg)"]                                          &  & H(b) \arrow[d, "H(g)"']                                \\
                                                              & H(y^Lx^Lax^{RR}y^{RR}) \arrow[rr, "{(\gamma_{H,y^R})_{y^Lx^Lax^{RR}}}"] \arrow[rrdd, "{(\gamma_{H,y^Rx^R})_{y^Lx^La}}"'] &  & H(yy^Lx^Lax^{RR}) \arrow[rr, "H(\coev_y x^Lax^{RR})"] \arrow[dd, "{(\gamma_{H,x^R})_{yy^Lx^La}}"] &  & H(x^Lax^{RR}) \arrow[dd, "{(\gamma_{H,x^R})_{x^La}}"'] \\
                                                              &                                                                                                                            &  &                                                                                                    &  &                                                        \\
                                                              &                                                                                                                            &  & H(xyy^Lx^La) \arrow[rrd, "H(\coev_{xy} a)"'] \arrow[rr, "H(x\,\coev_y x^L a)"]                       &  & H(xx^La) \arrow[d, "H(\coev_x a)"']                    \\
                                                              &                                                                                                                            &  &                                                                                                    &  & H(a)                                                  
\end{tikzcd}
\]
To save space, we used shorthand notations: $y^Lgy^{RR}$ stands for $\id_{y^L}\ot g\ot \id_{y^{RR}}$ and similar for others.
Natural transformations in $\coc\cC$ commuting with $\gamma$ are automatically natural transformations in $\coc\xc$. The assignment 
\begin{align*}
Z(\cC\hookrightarrow\coc\cC)&\to \coc\xc\\
(H,\gamma_{H,-})&\mapsto \hat H
\end{align*}
is then a functor.

It is straightforward to check that the above two constructions are inverse to each other, and thus $\coc\xc$ is equivalent to $Z(\cC\hookrightarrow\coc\cC)$.
\end{proof}
\begin{corollary}
    Since $Z(\cC\hookrightarrow\coc\cC)$ has a canonical braided monoidal structure, so is $\coc\xc$. The composition $Z(\cC)\xrightarrow{Y} \coc\xc\cong Z(\cC\hookrightarrow\coc\cC)\hookrightarrow Z(\coc\cC)$ is a braided fully faithful functor. It is clear that $Z(\cC)\xrightarrow[]{Y}\coc\xc\cong Z(\cC\hookrightarrow\coc\cC)$ coincides with the natural inclusion $Z(\cC)\hookrightarrow Z(\cC\hookrightarrow\coc\cC)$, i.e., restricting to the full subcategory of $Z(\cC\hookrightarrow\coc\cC)$ on objects $(H,\gamma_{H,-})$ where $H$ is representable. Pullback by $L$ manifestly coincides with the forgetful functor 
    \begin{align*}
    &\left(Z(\cC\hookrightarrow\coc\cC)\to\coc\cC\right) \\&=\left(Z(\cC\hookrightarrow\coc\cC)\cong \coc\xc \xrightarrow{-\circ L} \coc\cC\right).
    \end{align*}
\end{corollary}

\begin{remark}
    We believe that the monoidal structure of $\coc\xc$ physically correspond to the composition of MPOs, i.e., tensor contraction along the time direction, although the physical bonds are not explicit in $\xc$ (they are formally ``integrated out'').
\end{remark}

\begin{remark}
    Note that the Drinfeld center and relative center can be equivalently defined by the bimodule functors
    \[ Z(\cC)\cong \Fun_{\cC|\cC}(\cC,\cC),\quad Z(F:\cC\to \cD)\cong \Fun_{\cC|\cC}(\cC,\cD).\]
    Indeed, our approach in this paper can be readily generalized to module functors. Let $\cA$ be a monoidal $\cV$-category and $\cM,\cN$ left module $\cV$-categories over $\cA$; we can define a $\cV$-category $\hc{\cA}{\cM}{\cN}$ in a similar way as the tube category $\xc$:
    \begin{itemize}
        \item Objects are $\Ob(\hc{\cA}{\cM}{\cN})=\Ob(\cN)\times \Ob(\cM)$;
        \item Morphisms are defined by the coend: $$\hc{\cA}{\cM}{\cN}((r,x),(s,y)):=\int^{a\in \cA} \cN(r,as)\otv \cM(ay,x);$$
        \item Composition and identity are defined similarly as the tube category.
    \end{itemize}
    Such category $\hc{\cA}{\cM}{\cN}$ might be called the \emph{boundary tube category}, since it is the horizontal categorification of the weak Hopf algebra introduced in \cite{KK1104.5047} whose representations correspond to excitations on gapped boundaries of Levin-Wen models.
    We can similarly prove that
    \[ \coc{(\hc{\cA}{\cM}{\cN})}\cong \Fun^\text{lax}_\cA(\cM,\coc\cN),\]
    where $\Fun^\text{lax}_\cA$ denotes the category of lax left $\cA$-module $\cV$-functors (see, e.g.~\cite{KYZZ2104.03121,DSS1406.4204}), and the left $\cA$-module structure on $\coc\cN$ is similarly defined by the Day convolution. The details will be laid out in our future work. Note that when $\cA$ has left duals, lax left $\cA$-module functors are automatically strong left $\cA$-module functors~\cite[Lemma 2.10 in arXiv v3]{DSS1406.4204}. For a rigid monoidal $\cV$-category $\cC$, take $\cA=\cC\otv \cC^\rev$ ($\cC^\rev$ denotes the same category $\cC$ equipped with reversed tensor product) and $\cM=\cN=\cC$, we have
    \[ \coc{(\hc{\cC\otv\cC^\rev}{\cC}{\cC})}\cong \Fun_{\cC|\cC}(\cC,\coc\cC)\cong Z(\cC\hookrightarrow\coc\cC)\cong\coc\xc.\] 
    $\hc{\cC\otv\cC^\rev}{\cC}{\cC}$ is the tube category on the boundary $\cC$ of $\cC\otv \cC^\rev$, equivalently on the trivial domain wall $\cC$ between $\cC$ and $\cC$ itself by the folding trick. We may also view $\xc$ as the tube category with a single leg in or out (an object in $\xc$ is an object in $\cC$), and $\hc{\cC\otv\cC^\rev}{\cC}{\cC}$ is then the tube category with a double leg in or out (an object in $\hc{\cC\otv\cC^\rev}{\cC}{\cC}$ is a pair of objects in $\cC$). $\hc{\cC\otv\cC^\rev}{\cC}{\cC}$ and $\xc$ are Morita equivalent in the sense that their categories of representations are equivalent.
\end{remark}

\acknowledgments
TL is grateful to Tin Hay (Giovanni) Leung, Xiao-Gang Wen and Zhengcheng Gu for discussions. This work is supported by start-up funding from The Chinese University of Hong Kong, and by funding from Research Grants Council, University Grants Committee of Hong Kong (ECS No.~24304722).




\bibliographystyle{JHEP}
\bibliography{biblio.bib}


\end{document}